\documentclass[12pt]{article}
\def\e{{\rm e}}
\def\a{\alpha}

\usepackage[english]{babel}
\usepackage{amssymb}
\usepackage{amsmath}
\usepackage{indentfirst}
\usepackage{graphicx}
\usepackage{graphics}
\usepackage{epsfig}
\usepackage{subfigure}
\usepackage{multirow}
\usepackage{color}
\usepackage{subfigure,a4}
\usepackage{multirow}
\usepackage{epsfig,ulem}
\usepackage{dcolumn}
\usepackage{bm}

\newcommand{\be}{\begin{eqnarray}}
\newcommand{\ee}{\end{eqnarray}}

\usepackage[english]{babel}
\usepackage{amssymb}
\usepackage{amsmath}
\usepackage{amsfonts}
\usepackage{amsbsy}
\usepackage{indentfirst}
\usepackage{graphicx}
\usepackage{color}

\newcommand{\C}{\mathbb{C}}
\newcommand{\R}{\mathbb{R}}

\def\e{{\rm e}}

\newtheorem{theorem}{Theorem}[section]
\newtheorem{pro}{Proposition}[section]
\newtheorem{rmk}{Remark}[section]

\newenvironment{proof}[1][Proof]{\noindent\textbf{#1.} }{\ \rule{0.5em}{0.5em}}
\def\d{{\rm d}}

\def\a{\alpha}
\def\diag{{\rm diag}}
\def\<{\langle}
\def\>{\rangle}

\def\diag{{\rm diag}}

\newcommand{\beq}{\begin{equation}}
\newcommand{\eeq}{\end{equation}}

\newcommand{\bmat}{\begin{displaymath}}
\newcommand{\emat}{\end{displaymath}}

\def\1{{\bf 1}}

\def\ga{\gamma}
\begin{document}
\title{
Non-Hermitian
quantum mechanics of bosonic operators}
\author{
N. Bebiano\footnote{ CMUC, University of Coimbra, Department of
Mathematics, P 3001-454 Coimbra, Portugal (bebiano@mat.uc.pt)},
J.~da Provid\^encia\footnote{CFisUC, Department of Physics,
University of Coimbra, P 3004-516 Coimbra, Portugal
(providencia@teor.fis.uc.pt)}~
and J.P. da
Provid\^encia\footnote{Department of Physics, Univ. of Beira
Interior, P-6201-001 Covilh\~a, Portugal
(joaodaprovidencia@daad-alumni.de)}
} \maketitle

\begin{abstract}
The spectral analysis of a non-Hermitian unbounded
operator appearing in quantum physics is our main concern.
The properties of such an operator are essentially different
from those of Hermitian Hamiltonians, namely due to spectral
instabilities.
We demonstrate that the considered operator
and its adjoint can be
diagonalized when expressed in terms of certain
conveniently constructed operators. We show that
their eigenfunctions constitute complete systems,
but do not form Riesz bases. Attempts to overcome this difficulty
in the quantum mechanical set up are pointed out.
\end{abstract}
\section{Introduction}
The main motivation for this article is the following.
In conventional formulations of
non-relativistic quantum mechanics,
the Hamiltonian operator is Hermitian (synonymously, self-adjoint), and so
has real eigenvalues and an orthonormal set of eigenfunctions.
These fundamental issues are in the
heart of von Neumann quantum paradigm for physical observability and dynamical
evolution. Certain relativistic extensions of quantum
mechanics lead to non-Hermitian Hamiltonian operators,   $H\neq
H^*$, for $H^*$ the adjoint of $H$ (e.g. see \cite[Chapter VIII]{davydov}).
Extending the set of allowed operators by including $\cal PT$-symmetric ones,
$\cal P$ being the {\it reflection} operator, ${\cal P}f(x)=f(-x)$ and $\cal T$
the {\it time reflection} operator, ${\cal T}f(x)=\overline{f(x)}$,
yields spectra with reflection symmetry with respect to the real axis.
\color{black} In this context, non-Hermitian Hamiltonians
with a purely discrete spectrum have been the object of intense
research activity
\cite{bagarello,*,[1],[2],mostafa,providencia,scholtz,[3]},
attempting to build quantum mechanical theories with physical observables
described by these operators.
The subtleties of non self-adjoint Hamiltonians deserved the attention of physicists
and mathematitians, as they may originate new and unexpected phenomena, and the
non self-adjoint theory has revealed difficulties and challenging problems.

We focus on the spectral analysis of non-Hermitian operators,
a field with applications in many areas of physics, most remarkably in
quantum mechanics.
We summarise theoretical classical and recent developments on this topic, and
the theory is illustrated with a concrete example.
We explicitly determine the eigenfunctions and
eigenvalues of a non-Hermitian
operator describing two interacting bosons, concretly, a non self-adjoint 2D
harmonic oscillator.

Let $\cal H$ be an infinite dimensional separable Hilbert space
endowed with the inner product $\langle\cdot,\cdot\rangle$, and corresponding norm $\|\cdot\|$.
For $H$ a non-Hermitian operator in $\cal H$ with a purely discrete
spectrum, the strategy of finding a Hermitian operator $H_0$ and an
invertible operator $P$ such that
\begin{equation}\label{**}H=PH_0P^{-1},\end{equation}
has been exploited.
This idea refers that such a non-Hermitian $H$
can be viewed essentially as an alternative
representation of a Hermitian operator with the same eigenvalues and
multiplicities, whenever $P,~P^{-1}$ are both bounded (see Ref. \cite{krejcirik}).

The relation (\ref{**}) is closely related to the {\it
quasi-Hermicity} of $H$
\begin{equation}\label{***}QH=H^*Q,
\end{equation}
where $Q=(P^{-1})^*P^{-1}$ is a positive definite operator called a {\it
metric operator}. (The relation (\ref{***}) means that $QH$ and $H^*Q$
have the same domain and act in the same manner). In  fact, $H$ with property (\ref{**}), is
formally Hermitian with respect to the modified inner product
$\langle Q\,\cdot,\cdot\rangle,$
$$\langle QHx,y\rangle=\langle x,H^*Qy\rangle=\langle x,QHy\rangle, ~~x,y\in\cal H .$$
The metric is said to be {\it non singular} if it is bounded,
invertible and boundedly invertible, otherwise it is {\it
singular}.  
The pathologies of
non-Hermitian operators with singular metric may be serious,
since the metric can transform a basis into a set without any
reasonable basicity properties. Moreover, the spectral stability
with respect to small perturbations is not ensured and complex eigenvalues
can appear distant to the original ones \cite{siegle}.

The rest of this note is organized as follows. In Section \ref{S2}
we recall some issues used subsequently. In Section \ref{S3}
we investigate spectral properties of an unbounded
non-Hermitian operator acting in $\cal H$, quadratic in a pair of
bosonic operators, and of their adjoints. It is shown that the
obtained eigenfunctions are complete systems but do  not form Riesz
bases for $\cal H$, and so the validity of useful properties of
Hermitian operators is not   guaranteed.
In fact, non-Hermitian operators behave in an  essentially
different way, due in particular to spectral instabilities and to the lack of
basicity  of the Hamiltonian eigenfunctions.
In Section \ref{S5} we
discuss incidences of non-Hermiticity in quantum physics.

\section{Preliminaries}\label{S2}
In the sequel, we shall be concerned with ${\cal
H}=L^2(\R^2)$,  the Hilbert space of square-integrable complex
valued functions in two real variables $x,y$, equipped with the
inner product
\begin{equation}\langle\Phi,\Psi\rangle=\int_{-\infty}^{+\infty}\int_{-\infty}^{+\infty}{\Phi(x,y)}\overline{\Psi(x,y)}\d
x\d y,\label{innerproduct}\end{equation}
and corresponding norm $\|\cdot\|.$
\color{black}
Many important operators in quantum physics are unbounded, which restricts
their domains of definition to adequate subsets of the Hilbert space
where they live. For instance,
these domains are considered to be dense.

The set of eigenfunctions of $H$, $\{\Psi_n\}_{n=1}^\infty$, forms a {\it complete} family in $\cal H$
if the span of $\Psi_n$ is dense in $L^2(\R^2)$, or, equivalently, the
orthogonal complement of this linear span in the Hilbert space reduces to the zero function.

Eigenfunctions of non-Hermitian operators are in general not orthogonal or even
do not form a complete family. It should be stressed that the completeness of an orthogonal
set  $\{\Psi_n\}_{n=1}^\infty$, adequately normalized, does not imply,
by its own, basicity. An orthonormal set $\{\Psi_n\}_{n=1}^\infty$
does not guarantee that any $\psi\in\cal H$
admits a unique expansion of the form
$$\psi=\sum_{n=1}^\infty c_n\Psi_n,$$
where, in this event, it should be $c_n=\langle\psi,\Psi_n \rangle,~n\geq1.$

The concept of eigenbasis is very important in quantum mechanics. In contrast to
the case of eigenfunctions of Hermitian operators, which are basis,
for non-Hermitian operators this property fails
in general, and so the notion of {\it Riesz basis} may be of interest.
We say that $\{\Psi_n\}_{n=1}^\infty$ is a Riesz basis if
for any $\psi\in \cal H$, there exists a positive constant $C$ independent of $\psi$ such that
$$C^{-1}\|\psi\|^2\leq\sum_{n=1}^\infty|\langle\psi,\Psi_n\rangle|^2\leq C\|\psi\|^2.$$

The eigenfunctions of an operator $H$ with purely discrete spectrum constitute a
Riesz basis if and only if $H$ is quasi-Hermitian via a nonsingular
bounded metric $Q$ (see \cite{bagarello*}). The eigenfunctions of a non-Hermitian
$H$, despite possibly being
a complete family, may not form a Riesz basis, as occurs frequently
with several models \cite{krejcirik}. It should be noticed that
Riesz basicity, like the spectrum, is not preserved by unbounded metrics.

We introduce, as a first example of familiar operators
in quantum mechanics, the {\it multiplication} operators $x$ and
$y$
$$f(x,y)\rightarrow{x f(x,y)},\quad f(x,y)\rightarrow{y f(x,y)},$$
defined in their maximal domains
${\cal D}(x)=\{\psi\in L^2(\R^2):x\psi\in L^2(\R^2)\}$ and
${\cal D}(y)=\{\psi\in L^2(\R^2):y\psi\in L^2(\R^2)\}$.

We consider the partial differential operators ${\partial/\partial
x}$ {and} ${\partial/\partial
y}$ defined by
$$f(x,y)\rightarrow{\partial f(x,y)\over\partial x},\quad
f(x,y)\rightarrow{\partial f(x,y)\over\partial y},$$
with domains
${\cal D}(\partial/\partial x)=W^{1,2}(\R^2)$ and
${\cal D}(\partial/\partial y)=W^{1,2}(\R^2),$
where $W^{1,2}(\R^2)$ stands for
the usual Lebesgue space of functions in $L^2(\R^2)$, whose first
partial derivatives belong to $L^2(\R^2)$.
These operators are known as {\it momentum} operators.
The partial differential operators ${\partial^2/\partial
x^2}$ {and} ${\partial^2/\partial
y^2}$ are defined by
$$f(x,y)\rightarrow{\partial^2 f(x,y)\over\partial x^2},\quad
f(x,y)\rightarrow{\partial^2 f(x,y)\over\partial y^2},$$
with domains
${\cal D}(\partial^2/\partial x^2)=W^{2,2}(\R^2)$ and
${\cal D}(\partial^2/\partial y)=W^{2,2}(\R^2),$ the Lebesgue space
of functions in $L^2(\R^2)$ whose first and second derivatives belong
to $L^2(\R^2)$.

The {\it (standard) bosonic operators}
$$a=:x+{1\over2}{\partial\over\partial x},\quad
b=:y+{1\over2}{\partial\over\partial y},$$
with domains
$${\cal D}(a)={\cal D}(b)=\{\Psi\in W^{2,2}(\R^2):x\Psi, y\Psi\in L^2(\R^2)\}$$
are useful in our discussion.
The operators $a$ and $b$ are known to be densely defined, closed, and their adjoints read
$$a^*=x-{1\over2}{\partial\over\partial x}, \quad b^*=y-{1\over2}{\partial\over\partial y},$$
We recall that, conventionally, $a,b$ are said to be {\it
annihilation operators}, while $a^*,b^*$ are  {\it creation
operators}. It is worth noticing that these operators are unbounded,
and they satisfy the commutation rules (CR's),
\begin{equation}\label{CR1}[a,a^*]=[b,b^*]=\1,\end{equation}where $\1$ is the identity operator on
$L^2(\R^2).$ (This means that $aa^*f-a^*af=bb^*f-b^*bf=f$ for any $f$
in ${\cal D}(a)$ and ${\cal D}(b)$) Furthermore,
\begin{equation}\label{CR2}[a,b^*]=[b,a^*]=[a^*,b^*]=[a,b]=0.\end{equation}
As it is well-known, the canonical commutation relations (\ref{CR1})
and (\ref{CR2}) characterize an algebra of Weil-Heisenberg (W-H).
Moreover, the following holds,
$$a\Phi_0=b\Phi_0=0,$$
for $\Phi_0=\e^{-(x^2+y^2)}
$ in ${\cal D}(a)$ and ${\cal D}(b)$, a so-called {\it vacuum
state}.
The set of functions
\begin{equation}\label{Phimn}\{\Phi_{m,n}=a^{*m}b^{*n}\Phi_0:~m,n\geq0\},\end{equation}constitutes a
basis of $\cal H$, that is, every vector in $L^2(\R^2)$ can be
uniquely expressed in terms of this system, which is {\it
complete}, since 0 is the only vector orthogonal to all its
elements.

\section{Non self-adjoint Hamiltonian $H$
describing two interacting bosonic operators}
\label{S3}
\subsection{Model}
We consider the operator in $L^2(\R^2)$,
\begin{equation}\label{Ham}H:=-{1\over4}\left({\partial^2\over\partial x^2}+{\partial^2\over\partial y^2}\right)
+\gamma\left(y{\partial\over\partial x}+x{\partial\over\partial y}\right)
+(x^2+y^2),\quad \gamma\in\R,\end{equation}
which is obviously non self-adjoint for $\gamma\neq0.$
The closedness of the operator is the essential starting point
for the study of its spectrum. In the perspective of a convenient domain of $H$,
aiming at defining ${\cal D}(H),$ we regard the cofactor of $\gamma$ in the non-self-adjoint term,
$$V=\left(y{\partial\over\partial x}+x{\partial\over\partial y}\right),
$$ as a perturbation of the 2D harmonic oscillator,
$$\Re(H)=-{1\over4}\left({\partial^2\over\partial x^2}+{\partial^2\over\partial y^2}\right)
+(x^2+y^2).$$ It is easily seen that
this term is relatively bounded with the relative bound less
than one with respect to $\Re(H)$, provided that $\gamma$ is less
than one. Under this condition, $H$ is
a closed operator on the domain of the harmonic oscillator:
$$W^{2,2}(\R^2) \cap L^2(\R^2,(x^4+y^4) \d x \d y).$$ Here, we have used
the standard perturbation result that states the following.
If $\Re(H)$ is closed and $V$ is relatively bounded with respect to
 $\Re(H)$, with the relative bound smaller than 1, then  $\Re(H)+\lambda V$, with $\lambda<1$,
 is closed (\cite[Theorem 3.3]{Krejcirik*}).

We notice that ${\cal D}(H)$ contains the subspace $\mathfrak{S}(\R^2)$
constituted by functions
$f(x,y)$ such that $e^{-2\gamma xy}f(x,y)\in \mathfrak{S}(\R^2)$
and, in turn, this one
contains $C_0^\infty(\R^2)$.


In terms of creation and annihilation operators, $H$ is equivalently defined as
\begin{equation}\label{H}H=a^*a+bb^*+\gamma(a^*b^*-ab),\quad \gamma\in\R,\end{equation}
and so $H$ is quadratic in the bosonic operators $a,b$ and in
their adjoints.
\subsection{Spectrum}
As it is well known,
the {\it resolvent set} of $H$, denoted by $\rho(H)$, is constituted
by all the complex numbers for which the {\it resolvent operator}
$\lambda\in\rho(H)\rightarrow(H-\lambda)^{-1}$
exists as a bounded operator on $\cal H$. The complement
$$\sigma(H)=\C\backslash\rho(H)$$
is called the {\it spectrum} of $H$.
The set of all eigenvalues of $H$ is called the
{\it point spectrum}, 
denoted by $\sigma_p(H),$ and
formed by complex numbers $\lambda$ for which
$H-\lambda:{\cal D}(H)\rightarrow {\cal H}$ is not injective.
The spectrum of an operator on a finite dimensional
Hilbert space is exhausted by the eigenvalues, but, in the infinite dimensional setting,
there are additional parts to be considered.
 Those $\lambda$ which are not
eigenvalues but $H-\lambda$ is not
bijective constitute the {\it continuous} or {\it residual
spectrum}, depending on the range ${\rm Ran}(H-\lambda)$ being, respectively, dense
or not. The spectrum  $\sigma(H)$ is the union of
these three disjoint spectra.
The  spectrum of self-adjoint operators is nonempty, real,
and the residual spectrum is empty, while the spectrum of non
self-adjoint operators can be empty or coincide with the whole
complex plane (see, e.g., refs. \cite{siegle,[3]}).

\section{Accretivity}
It may be advantageous to estimate the spectrum in terms of the {\it numerical range}:
$$W(H):=\{\langle H\psi,\psi\rangle:\psi\in{\cal D}(H),\|\psi\|=1\}.$$
In general $W(H)$ is neither open nor closed, even when $H$ is a closed
 operator. However, it is always convex and  for $H$ bounded,
 $$\sigma(H)\subset\overline{W(H)}.$$
\begin{pro}\label{P3.1}
The numerical range of $H$ is bounded by the hyperbola
$$y^2+\gamma^2 ({1-x^2})=0,\quad x\geq1.$$
\end{pro}
\begin{proof}
The numerical range of $H$ is determined as follows.
Recalling that $W(H)$ is convex, let us consider the supporting line of $W(H)$
perpendicular to the direction $\theta$. Recall that the distance
of this line to the origin is the lowest eigenvalue of
$$\Re (\e^{-i\theta}H)=(a^*a+b^*b)\cos\theta-i\gamma(a^*b^*-ab)\sin\theta,$$
provided this operator is bounded from below, which occurs for $-\pi/2\leq\theta<\pi/2$.
The eigenvalues of  $\Re(\e^{-i\theta}H)$  are readily determined by
the EMM \cite{davydov}, and they are readily found to be
$$E_\theta(1+n+m),\quad n,m\geq 0,$$
where
$$E_\theta={\sqrt{1 - \gamma^2 + \cos(2 \theta) + \gamma^2 \cos(2 \theta)}\over\sqrt2}.$$
Thus, the supporting line under consideration is given by
\begin{equation}x\cos\theta+y\sin\theta=E_\theta.\label{support}\end{equation}
As it is well known, the equation of the boundary of $W(H)$ is found eliminating $\theta$
between (\ref{support}) and
$$-x\sin\theta+y\cos\theta=E'_\theta,
$$
being readily obtained as the branch of  hyperbola
$$y=\pm\gamma\sqrt{x^2-1}, ~ x\geq1.$$
\end{proof}\color{black}

 {\it Sectorial} operators are defined by the property
that their numerical range is the subset of  a sector
$$S_{w,\theta}=\{z\in\C:|\arg{(z-w)}|\leq\theta\}$$
with $w\in\R$ and $0\leq\theta<\pi/2,$ called, respectively, the {\it vertex} and {\it semi-angle} of
$H$. An operator is said to be {\it accretive} if the  vertex can be chosen at the origin, i.e.,
$W(H)\subset S_{0,\pi/2}$. An operator $H$ is $m$-{\it accretive} if
its numerical range is contained in the right closed half-plane
and the {\it resolvent bound} for any $\lambda$ with $\Re\lambda<0$, holds:
$$\forall\lambda\in\C,~\Re\lambda<0,~\|(H-\lambda)^{-1}\|\leq1/|\Re \lambda|.$$
Any $m$-accretive operator is closed and densely defined (\cite{bagarello*},p.251)
\begin{pro}
The operator $H$ is $m$-accretive.
\end{pro}
\begin{proof}
We show that for any $z\in\C,$ with $\Re z<0$, the resolvent bound holds.
We have
$${\rm dist}(z,\overline{W(H)})\leq|\langle H\psi,\psi\rangle-z|=
|\langle(H-z)\psi,\psi\rangle|\leq\|(H-z)\psi\|.$$
As ${\rm dist}(z,\overline{W(H)})\geq|\Re z|,$ the result follows,
having in mind Proposition \ref{P3.1}.
\end{proof}

Any $m$-accretive operator is closed and densely defined.
In fact, by Proposition 5.2.1 in \cite{bagarello*} p. 246 if $H$ is not closed, then $\sigma(H)=\C$.

A closed operator $H$ in $\cal H$  has a {\it compact resolvent} if
$\rho(H)\neq{\O}$ and $(H-\lambda)^{-1}$, for some $\lambda\in\rho(H),$
is a compact operator.

It is known that \cite[Theorem IX, 2.3]{Krejcirik*},
if $H$ has a compact resolvent, then $\sigma(H)=\sigma_p(H)$.
The operator $H$ has a compact resolvent, as $\Re(H)$ is an $m$-accretive
operator (since $\Re(H)$ is Hermitian and $W(\Re(H))$ lies on the positive real axis)
with compact resolvent and, moreover, $V$ is relatively bounded with respect to $\Re(H)$
with relative bound smaller than 1. Then, $\Re(H)+\lambda V$ has a
compact resolvent \cite[Theorem 5.4.1]{Krejcirik*}.
\subsection{Eigenvectors and eigenvalues of $H$}\label{ss3.1}
To obtain the eigenvalues of $H$, we firstly consider the
selfadjoint operator in $L^2(\R^2),$
\begin{equation}\label{H0}H_0=-{1\over4}\left({\partial^2\over\partial x^2}+{\partial^2\over\partial y^2}\right)
+(1+\gamma^2)(x^2+y^2),~~\gamma<1.\end{equation}
We know {\it a priori} that the eigenfunctions of $H_0$ (after normalization)
form an orthonormal family in $L^2(\R^2)$  and the corresponding eigenvalues are real.
The eigenvalues of $H_0$ are easily determined as  $$E_{mn}=(1+m+n)\sqrt{1+\gamma^2},~~m,n\geq0,$$
and the associated eigenfunctions $\Phi_{mn}$  are 
\begin{equation}\Phi_{m,n}=KH_m((1+\gamma^2)^{1/4}x)\e^{-\sqrt{1+\gamma^2}x^2}
\times H_n((1+\gamma^2)^{1/4}y)\e^{-\sqrt{1+\gamma^2}y^2},\label{eigenf}\end{equation}
where $K=(2(1+\gamma^2)/\pi)^{1/2}$ and $H_n(t)$ is the $n$th Hermite polynomial in $t$.
That is, $\Phi_{mn}$ is factorized as follows
$$\Phi_{mn}(x,y)=\Phi_m(x)\Phi_n(y),$$ where $\Phi_n(t)$ are the
usual eigenfunctions of the famous harmonic oscillator \cite{bagarello*}.
For the sake of completeness, we show how to obtain $\Phi_{mn}$.
Indeed, let us consider the differential operators
$$g={1\over2(1+\gamma^2)^{1/4}}~{\partial \over\partial x}+(1+\gamma^2)^{1/4}x,\quad
h={1\over2(1+\gamma^2)^{1/4}}~{\partial \over\partial y}+(1+\gamma^2)^{1/4}y,$$
$$g^*=-{1\over2(1+\gamma^2)^{1/4}}~{\partial \over\partial x}+(1+\gamma^2)^{1/4}x,\quad
h^*=-{1\over2(1+\gamma^2)^{1/4}}~{\partial \over\partial y}+(1+\gamma^2)^{1/4}y,$$
which
satisfy the Weil-Heisenberg
commutation rules,
$$[g,g^*]=[h,h^*]=\1, [g,h]=[g^*,h^*]=[g^*,h]=[g,h^*]=0.$$
We easily find that
$$H_0=\sqrt{1+\gamma^2}(g^*g+h^*h+\1).$$
The {\it groundstate} eigenfunction of $H_0$, which is the vacuum of the operators $g,h$, is
$$\Phi_0(x,y)=\kappa_{0,0}\e^{-\sqrt{1+\gamma^2}(x^2+y^2)},$$
being the remaining eigenfunctions,
\begin{equation}\nonumber\Phi_{mn}(x,y)=\kappa_{mn}g^{*m}h^{*n}\Phi_0(x,y),\end{equation}
where $\kappa_{m,n}$ are normalization factors.

The  eigenfunctions $\Phi_{mn}(x,y)$ constitute a complete set in $L^2(\R^2)$, as
can be easily verified by adapting the
standard proof of completeness of Hermite functions.
The functions are orthogonal, and the
system forms a basis of $L^2(\R^2)$. Details can be found, for
instance, in \cite{davydov,dirac}.

To find the eigenvalues of $H$ we notice that $H$ is {\it formally} similar to $H_0$
\begin{equation}\label{sim}H_0=\e^{-2\gamma xy}H\e^{2\gamma xy}.\end{equation}
The word ``formally" refers to the fact that the operator $P=\e^{2\gamma xy}$ is unbounded.
Nevertheless, the similarity relation is well defined on the eigenfunctions of $H_0$.

For $\Phi(x,y)$ an arbitrary differentiable function in the domain of $H_0$,
we have
\begin{eqnarray*}&&\left(-{1\over4}\left({\partial^2\over\partial x^2}+{\partial^2\over\partial y^2}\right)
+\gamma\left(y{\partial\over\partial x}+x{\partial\over\partial y}\right)
+(x^2+y^2)\right)\e^{2\gamma xy}\Phi(x,y)\\&&=\e^{2\gamma xy}\left(-{1\over4}\left({\partial^2\over\partial x^2}+{\partial^2\over\partial y^2}\right)
+(1+\gamma^2)(x^2+y^2)\right)\Phi(x,y),
\end{eqnarray*}
that is
$$H\e^{2\gamma xy}\Phi=\e^{2\gamma xy}H_0\Phi,$$
and the operator equality
$$H\e^{2\gamma xy}=\e^{2\gamma xy}H_0$$
holds.

Since $$H_0\Phi_{mn}=E_{mn}\Phi_{mn},$$
we easily get
$$H\e^{2\gamma xy}\Phi_{mn}=E_{mn}\e^{2\gamma xy}\Phi_{mn}.$$
The  eigenfunctions of $H$ are expressed as
$$\Psi_{mn}=\e^{2\gamma xy}\Phi_{mn},$$
and the eigenvalues of $H$ are those of $H_0$.
Thus, the eigenvalues of $H$
are real positive and $\Psi_{mn}\in{\cal D}(H)$,
as it should be.



\subsection{Eigenfunctions and eigenvalues of $H^*$}\label{ss3.3}
Following the procedure in Subsection \ref{ss3.1},
it can be shown that, formally, we may write,
$$H_0=\e^{2\gamma xy}H^*\e^{-2\gamma xy},$$
implying that eigenfunctions and associated eigenvalues of
$H^*$ are given respectively by
$$\tilde\Psi_{m,n}=\e^{-2\gamma xy}\Phi_{m,n}\in {\cal D}(H^*),\quad E_{m,n}=
(m+n+1)\sqrt{1+\gamma^2},~~ m,n\geq0 .$$
\subsection{Biorthogonality of the eigenfunctions} It is
straightforward to see that the set  $\{\Psi_{m,n}: m,n\geq0\}$ is
not constituted by orthogonal functions. However, the vector system
$\{\tilde\Psi_{m,n}
: m,n\geq0\},$
formed by the eigenfunctions of $H^*$, is biorthogonal to the
eigensystem of $H$, $\{\Psi_{m,n}: m,n\geq0\}$. Indeed, we have
\begin{eqnarray*}&&\langle\Psi_{m,n},\tilde\Psi_{p,q}\rangle=\langle
\e^{2\gamma xy}\Phi_{mn},\e^{-2\gamma xy}\Phi_{pq}\rangle=\langle
\Phi_{mn},\Phi_{pq}
\rangle=m!n!\delta_{mp}\delta_{n,q}
,\end{eqnarray*}
where $\delta_{ij}=1$ for $i=j$, otherwise $\delta_{ij}=0$,
represents the Kronecker symbol.

\subsection{Completeness of the eigenfunctions}
The completeness of both eigensystems of $H$ and $H^*$ can be shown to hold. The
operator is $m$-accretive as its numerical range lies a hyperbolical region
limited by the branch of hyperbola  $y=\sqrt{x^2-1},~x\geq1$.
Moreover, it can be easily seen that the
imaginary part of the resolvent of $-i H$ at $\delta<0$, is
non-negative Hermitian
$${1\over2i}\left({( -iH-\delta})^{-1}-{(iH^*-\delta)^{-1}}\right)\geq0.$$
As the resolvent is a trace class function, by the completeness
theorem \cite[Theorem VII.8.1]{gohberg}, we may conclude that the
eigenfunctions of $H$ form a complete system. Analogous arguments
are valid  for the eigensystem of $H^*$.\color{black}

We observe that completeness does not imply that any $\psi\in
L^2(\R^2)$ has a unique expansion
$$\psi=\sum_{m,n=0}^\infty c_{mn}\Psi_{mn},$$
a fundamental issue in quantum mechanics.

\subsection{Asymptotic behavior of the eigenfunctions}
\def\S{{\rm S}}  We
wish to discuss the asymptotic behavior of the Hamiltonian
eigenfunctions. To this end, it is convenient to introduce the
Planck constant $\hslash,$ explicitly. We also change the notation
slightly. We replace the notation $x,y$, used for the particle
coordinates, by $x_1,x_2$, and we denote the respective momenta by
$$p_1=-i\hslash{\partial\over\partial x_1},\quad p_2=-i\hslash{\partial\over\partial x_2}
.$$For simplicity, we also consider $\gamma=1,$ so that (\ref{H})
becomes
$$H={1\over4}(p_1^2+p_2^2)+x_1^2+x_2^2+i(x_1p_2+x_2p_1)-1.$$
Next, we introduce the change of variables, which constitutes a
canonical transformation,
$$X={1\over2}(x_1+x_2),~~x=x_1-x_2,~~P=p_1+p_2,~~p={1\over2}(p_1-p_2).$$
In terms of the new variables,  the Hamiltonian
is the sum of two summands, each one involving only one type
of variable, and becomes
$$H={1\over8}P^2+2X^2+iXP+{1\over2}p^2+{1\over2}x^2+{i}xp-1.$$
We observe that it is equivalent to consider
$m,n\rightarrow\infty$ or $\hslash\rightarrow0.$
We firstly concentrate on the summand
$$H_X={1\over8}P^2+2X^2+iXP.$$ In order to
characterize the asymptotic behavior of the eigenfunctions we use
the well known Wentzel-Kramers-Brillouin [WKB] approximation
\cite[Chapter III]{davydov} in leading order. Indeed,
we express the
{eigenfunction  of the energy operator} as
$$\Psi(X)=\e^{i \S(X)/\hslash}.$$
The function $\S(X)$, in leading order, is determined by the Jacobi
equation $${1\over8}\left({\d \S\over\d
X}\right)^2+2X^2+iX{\d \S\over\d X}=E,$$ where $E$
denotes the associated eigenvalue of the energy operator. Thus, we
readily obtain,
$${\d \S\over\d X}=2\sqrt{2 E-8 X^2}-4iX,$$and
$$S=X\sqrt{2 E-8X^2}+ {E\over\sqrt2} \arctan{2X\over\sqrt{ E-4X^2}}-{2}iX^2.$$
Similarly, the eigenfunction of $H^*$ is
$$\tilde\Psi(X)=\e^{i \tilde \S(X)/\hslash},$$with
$$\tilde \S=X\sqrt{2 E-8X^2}+ {E\over\sqrt2} \arctan{2X\over\sqrt{ E-4X^2}}+{2}iX^2.$$
Therefore, the integral
$$\int_{-\sqrt{E/2}}^{\sqrt{E/2}}\overline\Psi\Psi~\d X$$
approaches $+\infty$ as $\hslash\rightarrow0$, while the integral
$$\int_{-\sqrt{E/2}}^{\sqrt{E/2}}\overline{\tilde\Psi}\tilde\Psi~\d X$$ approaches $0$ as
$\hslash\rightarrow0$.

On the other hand,  the integral
$$\int_{-\sqrt{E/2}}^{\sqrt{E/2}}\overline{\Psi}\tilde\Psi~ \d X$$ remains finite as
$\hslash\rightarrow0$.

The summand
$$H_x={1\over2}p^2+{1\over2}x^2+{i}xp$$
may be similarly treated.

As a consequence, we get
$$\lim_{m,n\rightarrow\infty}\|\Psi_{mn}\|=\infty.$$
Thus, the
system of eigenvectors of $\{\Psi_{mn}\}$ does not form a Riesz
basis for $L^2(\R^2)$. Analogous conclusion holds for
$\{\tilde\Psi_{mn}\}$.


\subsection{Metric operator}
The existence of a positive definite $Q$ satisfying (\ref{***}) is equivalent
to the fact that the resolvent of $H$ satisfies
$$Q{( H-z_0)^{-1}}={(H^*-z_0)^{-1}}Q$$
for $z_0\in\rho(H)\cap\rho(H^*)\cap\R^2.$

As a consequence of the reality of the discrete spectrum of $H$ and of the
completeness of the corresponding eigenfunctions, it can be shown,
using the procedure in \cite{siegle}, that there exists a bounded
metric for $H$. The proof relies on the fact that the existence of
bounded metric for an unbounded $H$ can be transferred to the
same problem for its bounded resolvent.

Nevertheless,
there
does not exist a non singular metric operator ensuring
quasi-Hermiticity, because the eigenfunctions of $H$ do not
form a Riesz basis.  As a consequence, in the similarity of $H$ to a
self-adjoint operator, the basicity properties of these operators may be very different
(see \cite{bagarello} and \cite{bagarello*}). However, despite these
negative features, in the subspace ${\cal S}={\rm
span}\{\Psi_{m,n}\},$ a new inner product may be meaningfully
defined with the help of a metric operator (non-singular in $\cal
S$), with respect to which the restriction of $H$ to $\cal S$, say $h$,
should represent $H$.


\subsection{Physical Hilbert space}\label{ss3.9}
In order to specify the {\it physical Hilbert space} associated with the Hamiltonian (\ref{Ham}),
we chose its domain to be the function space
$${\rm Dom}(H)=\{\Psi:\e^{-2\gamma xy}\Psi\in W^{2,2}(\R^2) \cap L^2(\R^2,(x^4+y^4) \d x \d y)\},$$
endowed with the inner product,
$$\ll\Xi_1,\Xi_2\gg=\int_{-\infty}^\infty\d x\int_{-\infty}^\infty\d y
\e^{-4\gamma xy}\Xi_1(x,y)\overline{\Xi_2(x,y)},\quad \Xi_1,\Xi_2\in{\rm Dom}(H),$$
which involves the weight function $\e^{-4\gamma xy}.$
Since $H_0$ in (\ref{H0}) is a closed operator on the domain of the harmonic oscillator
$${\rm Dom}(H_0)=W^{2,2}(\R^2) \cap L^2(\R^2,(x^4+y^4) \d x \d y),$$
and
$$H\e^{2\gamma xy}\Phi=\e^{2\gamma xy}H_0\Phi,\quad\forall \Phi\in{\rm Dom}(H_0),$$
while
$$\e^{-2\gamma xy}H\Psi=H_0\e^{-2\gamma xy}\Psi,\quad\forall\Psi\in{\rm Dom}(H),$$
it follows that $H$ is a closed operator on ${\rm Dom}(H).$

Let $${\cal S}={\rm span}\{\Psi_{mn}\},\quad\tilde{\cal S}={\rm span}\{\tilde\Psi_{mn}\},$$
and ${\cal D}={\rm span}\{\Phi_{mn}\}.$

In terms of the linear operators
$\e^{2 \gamma xy}:{\cal D}\rightarrow{\cal S}$ and
$\e^{-2\gamma xy}:{\cal D}\rightarrow\tilde{\cal S}$ we may write
\begin{equation}\label{*}\Psi_{m,n}=\e^{2\gamma xy}\Phi_{m,n},\quad
\tilde\Psi_{m,n}=\e^{-2\gamma xy}\Phi_{m,n}.\end{equation}\color{black}
From (\ref{*}) it follows that
$${\cal S}=\e^{S}({\cal D}),~~\tilde{\cal S}=\e^{-S}({\cal D}).$$
Since
$$\langle\e^{-S}\Psi_{mn},\e^{-S}\Psi_{pq}\rangle=\delta_{mp}\delta_{nq},~S=2\gamma xy,$$
it follows that the eigenfunctions $\Psi_{mn}$ are orthogonal for this inner product,
$$\ll\Psi_{mn},\Psi_{pq}\gg=\delta_{mp}\delta_{nq}.$$

Next we show that {the linear space $\cal S$},
equipped with the inner product $\ll\cdot,\cdot\gg$, allows
the probabilistic interpretation of quantum mechanics.
{The symmetry is obviously satisfied.}
It may be pointed out that, from the point of view of
physics, only the action of the metric operator $\exp(-2S)$ on $\cal
S$ is significant and, in this event, it is unumbiguously defined.
We verify that $\exp(-2S)$ is a metric.

{For $\phi,\psi\in\cal S$,}
$\langle\exp(-2S)\phi,\psi\rangle$ is finite, and so the operator
$\exp(-2S)$ is bounded in $\cal S$. (Similarly, $\exp(2S)$ is bounded
in $\tilde{\cal S}$.) Moreover, $$\ll\psi,\psi\gg,~\psi\in{\cal S},$$
is nonnegative. (We notice that  $H$ leaves $\cal S$ invariant.)

The vectors $\Psi_{mn}$ are orthogonal with respect to
the new inner product
and any $\Psi\in\cal S$ may be expanded as a unique finite linear
combination of $\Psi_{nm},$
$$\Psi=\sum_{nm} c_{nm}\Psi_{nm},$$ with
$$\ll\Psi,\Psi_{nm}\gg=c_{nm}.$$ In order to
ensure the probabilistic interpretation of quantum mechanics, we
impose the normalization $\ll\Psi,\Psi\gg=1$, so that
$$\ll\Psi,\Psi\gg=\sum_{nm}|c_{nm}|^2=1.$$
The probability amplitude, is given by
$c_{nm},$ and  satisfies $\sum_{nm}|c_{nm}|^2=1.$


\section{Conclusions}\label{S5}We have initially considered
a non self-adjoint Hamiltonian $H$ whose eigenvalues and
corresponding eigenfunctions have been explicitly determined
The investigated Hamiltonian and (its adjoint) has real
eigenvalues and systems of biorthogonal eigenvectors. They  have
infinite diagonal matrix representations in the respective
eigensystems, which are complete. Nevertheless, they do
not form Riesz bases.

Viewing $H$ as the Hamiltonian of a physical model, problems arise
from non Hermiticity.
The original inner product defined in $\cal H$ is not adequate for
the physical interpretation of the model. A new $Q$-metric, which is
appropriate for that purpose, may be introduced (see Sub-Section
\ref{ss3.9}).
Following Mostafazadeh \cite{mostafa}, one can define a subspace of
the Hilbert space, and the restriction of the Hamiltonian operator
to that  subspace, so that it has the same spectrum and eigenfunctions
as the original one. The referred subspace remains invariant under
the action of $H$.
Remarkably, stating that this Hermitian operator represents in a
reasonable sense the non-Hermitian operator may be controversial,
since relevant information on the Hamiltonian may not be captured in the mentioned subspace.

Non-Hermitian operators have typically non-trivial pseudospectra.
It is known that the relation (\ref{**}) holds via a bounded and boundedly invertible
positive transformation if and only if (\ref{***}) holds with a positive bounded and
boundedly invertible metric \cite{krejcirik}. Further, if
(\ref{***}) holds with a positive bounded and boundedly invertible
metric, then the pseudospectrum of $H$ is trivial.
The concept of pseudospectrum is of great relevance for the description
of non-Hermitian operators in the context of quantum mechanics.
A non trivial pseudospectrum ensures the non existence of a bounded metric.
\section*{Acknowledgments} This work was partially
supported by the Centro de Matem\'atica da Universidade de Coimbra
(CMUC), funded by the European Regional Development Fund through the
program COMPETE and by the Portuguese
Government through the FCT - 
{Funda\c c\~ao} para a Ci\^encia e a Tecnologia under the project
PEst-C/MAT/UI0324/2011.


\begin{thebibliography}{99}


\bibitem{bagarello}  F. Bagarello, Construction of pseudo-bosons
systems, J. Math. Phys. {\bf 51} (2010) 023531.

\bibitem{bagarello1} F. Bagarello, More
mathematics on pseudo-bosons,  J. Math. Phys. {\bf 51} (2013)
063512.
\bibitem{bagarello*}F. Bagarello, J.-P. Gazeau,  F.H. Szafraniec, M. Znojil, Non-Selfadjoint Operators in Quantum Physics: Mathematical
Aspects, Wiley, 2015.
\bibitem{*}N. Bebiano, J. da Provid\^encia , J. P. da Provid\^encia, Mathematical
Aspects of Quantum Systems with a Pseudo-Hermitian Hamiltonian,
Brazilian Journal of Physics,{\bf 46} (2016) 152-156.
\bibitem{bebiano*}
N. Bebiano and J. da Provid\^encia, The EMM and the Spectral
Analysis of a Non Self-adjoint Hamiltonian on an Infinite
Dimensional Hilbert
Space, 
{\it Non-Hermitian Hamiltonians in Quantum Physics}. 
  {\bf 184} Springer Proceedings in Physics,   (2016)157-166 
\bibitem{[1]} C.M. Bender and S. Boettcher, Real Spectra in Non-Hermitian Hamiltonians Having PT Symmetry
, Phys. Rev. Lett., 80 (1998) 5243-5246.
\bibitem{bender} C.M. Bender, D.C. Brody and
H.F. Jones, Complex Extension of Quantum Mechanics,
 Phys. Rev. Lett, 89 (2002) 27041.
\bibitem{davies} E. Davies, Semi-Classical States for Non-Self-Adjoint Schr\"odinger Operators,
Commun. Math. Phys. {\bf 200} (1999)
 35-41.
 \bibitem{davydov}A.S. Davidov, Quantum Mechanics, 1965, Pergamon Press.
\bibitem{dirac}P.A.M. Dirac, The Principles of Quantum Mechanics,
 http: //www. fulviofrisone.com/attachments/article/447/Principles\%20of\%20Quantum \%20Mechanics\%20-\%20Dirac.pdf
\bibitem{gohberg} I. Gohberg, S. Goldberg and M. Kaashoek, Classes
of  minear operators, Birk\"auser Verlag, Basel, 1990, Vol.I.
\bibitem{[2]} A. Gonzal\'ez Lop\'ez and T. Tanaka, Nonlinear pseudo-supersymmetry in the framework of $N$-fold
supersymmetry,
 J. Phys. A: Math.
Gen. 39 (2006) 3715-23.
\bibitem{krejcirik} D. Krejcirik, P. Siegl, M. Tater, J. Viola, Pseudospectra
in non-Hermitian quantum mechanics, J. Math. Phys. {\bf 56} (2015)
103513.
\bibitem{Krejcirik*} D. Krejcirik and P. Siegl, Elements of spectral theory without the spectral theorem,
 in F. Bagarello, J.-P. Gazeau,  F.H. Szafraniec, M. Znojil, Non-Selfadjoint Operators in Quantum Physics: Mathematical
Aspects, Wiley, 2015.
\bibitem{mostafa}  A. Mostafazadeh, Pseudo-Hermitian Quantum Mechanics with Unbounded Metric
Operators, Cite as: arXiv:1203.6241 [math-ph],  Phil. Trans. R. Soc.
A {\bf 371} (2013) 20120050.
\bibitem{mostafa1}A.Mostafazadeh,
 Exact PT-symmetry is equivalent to Hermiticity, J.
Phys. A: Math. Gen. {36} (2003) 7081.
\bibitem{providencia} J. da Provid\^encia, N. Bebiano and JP. da
Provid\^encia, Non Hermitian operators with real spectra in Quantum
Mechanics, Brazilian Journal of Physics, 41 (2011) 78-85.
\bibitem{scholtz} F.G. Scholtz, H.B. Geyer and F.J.W. Hahne, Quasi-Hermitian operators in quantum mechanics and the variational
principle, Ann. Phys. NY 213 (1992) 74.
\bibitem{siegle}P. Siegl and D. Krejcirik,
The metric operator for the imaginary cubic oscillator, Phys. Rev. D
{\bf 86} (2012) 11702(R).
\bibitem{tanaka} T. Tanaka, Preprint quant-ph/0603075; T.
Tanaka, J. Phys. A. Math. Gen. 39 (2006) L369-L376.
\bibitem{[3]} M. Znojil, Should PT Symmetric Quantum Mechanics Be Interpreted as Nonlinear?, J. Nonlin. Math. Phys., 9 (2002) 122-123.
\end{thebibliography}
\end{document}